\documentclass[12pt]{article}

\usepackage[english]{babel}
\usepackage{amsfonts,psfrag}
\usepackage{amsmath}
\usepackage{amsthm}
\usepackage{graphicx}
\usepackage{tabularx}
\usepackage{listings}

\title{Classical versus Quantum Graph-based~Secret~Sharing}
\author{J\'er\^ome Javelle\textsuperscript{\small 1,}\footnote{Jerome.Javelle@imag.fr}~, Mehdi Mhalla\textsuperscript{\small 2,1,}\footnote{Mehdi.Mhalla@imag.fr}~, Simon Perdrix\textsuperscript{\small 2,1,}\footnote{Simon.Perdrix@imag.fr}}		
					
\date{\small \textsuperscript{1} LIG, University of Grenoble, France\\  \textsuperscript{2} CNRS }

\newtheorem{thm}{Theorem}
\newtheorem{lem}{Lemma}
\newtheorem{dfn}{Definition}
\newtheorem{property}{Property}

\begin{document}

\maketitle

\setlength{\parskip}{0.5em}

\begin{abstract}
We study a simple graph-based classical secret sharing scheme: every player's share consists of a random key together with the encryption of the secret with the keys of his neighbours. A characterisation of the authorised and forbidden sets of players is given. Moreover, we show that this protocol is equivalent to the  \emph{graph state quantum secret sharing} (GS-QSS) schemes \cite{MS08,KMMP,JMP11}  when the secret is classical. When the secret is an arbitrary quantum state, a set of players is authorised  for a GS-QSS  scheme if and only if, for the corresponding simple classical graph-based protocol,   the set is authorised and  its complement set is  not.

\end{abstract}

\section{Introduction}
Quantum secret sharing protocols \cite{CG99,G00} are quantum extensions of the classical secret sharing protocols \cite{B,Shamir}. They consists in encoding a secret into a multipartite quantum state. Each of the players of the protocol has a subpart of this quantum system, called a \emph{share}. Authorised sets of players are those that can recover collectively the secret. The encrypted secret can be a quantum state or a  classical message. 

In the literature, several quantum secret schemes have been introduced \cite{MS08}. In particular Markham and Sander have introduced QSS schemes based on graph states: the secret is encoded into a graph state, i.e. a quantum state which is characterised by a graph. Every vertex of the graph represent a player. Both classical and quantum secrets are considered in the graph state quantum secret sharing schemes. In \cite{MS08}, connections between the authorised sets for a classical secret and  the authorised sets for a quantum secret have been established. Recently, in \cite{JMP11},  a graphical characterisation of authorised and forbidden sets of players have been introduced in both cases of a classical and a quantum secret.

In this note, we study a  family of graph-based secret sharing protocols. Given a simple undirected graph and a given secret, every player's share is a pair which consists of a random key  together with the encryption of the secret by modular addition with the keys of the neighbour players in the graph (see the section \ref{sec:gss}).  %, where $s'=s$ if the player is a neighbour of $d$ and $s'=0$ otherwise.  
In section \ref{sec:gs-qss}, we show that the access structure  of the classical graph-based protocol coincides  with the access structure of the GS-QSS when the secret is classical. As a consequence, whenever the secret is classical, any  GS-QSS can be simulated by a simple classical scheme. 
Moreover, we point out the connections between the GS-QSS with quantum secret and the classical graph based protocols: the authorised sets of players are those which are authorised for the classical protocol for the same graph and its complement.

\section{Notations}

For a given classical or quantum secret sharing protocol over $n$ players, a subset of players is \emph{authorised} if the players of the subset can recover collectively the secret. A subset of players is \emph{forbidden} if they have no information about the secret. Notice that a third kind of sets of players may exist, those who have some partial information about the secret. The description of the authorised and forbidden sets is called the \emph{accessing structure} of the protocol. 

In this paper, the protocols are characterised by simple undirected\footnote{$G=(V,E)$ is a simple undirected graph if $\forall u\in V, (u,u)\notin E$ and  $(u,v)\in E\Rightarrow (v,u)\in E$} graphs. For a given graph $G=(V,E)$ and for any vertex $u\in V$,  $\mathcal N(u):=\{v\in V ~|~ (u,v)\in E\}$ denotes the neighbourhood of $u$ in $G$ ;  for any $D\subseteq V$, $Odd(D):=\{v\in V ~|~ |\mathcal N(v) \cap D| = 1 \mod 2\}$ is called the \emph{odd-neighbourhood} of $D$ in $G$. %On occasion use of the graph $G$ as subscript ($\mathcal N_G$, $Odd_G$) will avoid ambiguity.  
Notice that $Odd(D)=\bigtriangleup_{u\in D} \mathcal N(u)$, where $\triangle$ is the symmetric difference ($A\triangle B = (A\cup B)\setminus (A\cap B)$). 
For a given subset $B\subseteq V$, let $\overline B = V\setminus B$ be its complement. For a given graph $G=(V,E)$, let $\overline G= (V, \overline E\setminus \{(u,u) ~Ê|~u \in V\} )$  be its complement graph.

\section{A graph-based classical protocol}\label{sec:gss}

In this section, we consider a family of classical secret sharing protocols, each of these protocols is parameterised by a graph.  For a given simple undirected graph $G = (V, E )$, each vertex $i \in V$ represents a player. The secret to share is a bit $s\in \{0,1\}$. Each player receives the secret one-time padded by the keys of his neighbours. Formally, the classical graph-based secret sharing (GSS) protocol is defined as follows:

\paragraph{Sharing the secret.}
\begin{itemize}
\item For each player $i$, pick a bit $k_i$ uniformly at random in $\left\{0, 1\right\}$.
\item For each $i$, compute the value $ c_i = s + \sum_{{i'} \in \mathcal{N}( i)} k_{i'} \mod 2$.
\item Give player $i$ the couple $(k_i, c_i)$.
\end{itemize}

\noindent An example of GSS protocol is given in Figure \ref{fig:P5}. 

\begin{figure}[h]
\vspace{-1.5cm}
\psfrag{k1c1}[c][][1]{\small$(k_1,s+k_2)$}
\psfrag{k2c2}[c][][1]{\small$(k_2,s+k_1+k_3)$}
\psfrag{k3c3}[c][][1]{\small$(k_3,s+k_2+k_4)$}
\psfrag{k4c4}[c][][1]{\small$(k_4,s+k_3+k_5)$}
\psfrag{k5c5}[c][][1]{\small$(k_5,s+k_4)$}

		\centerline{\includegraphics[scale=0.5]{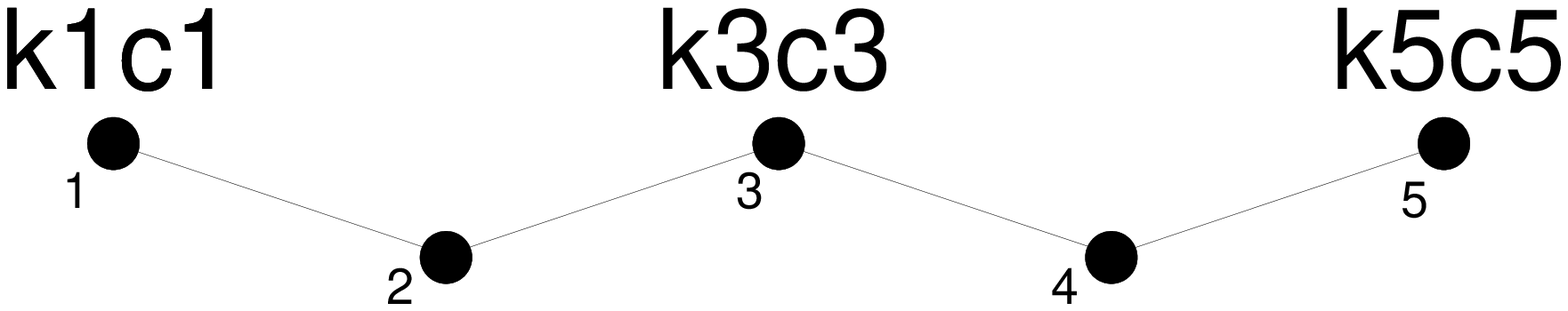} }

\vspace{-1.5cm}

\psfrag{k1c1}[c][][1]{\small$(0,0)$}
\psfrag{k2c2}[c][][1]{\small$(1,0)$}
\psfrag{k3c3}[c][][1]{\small$(1,0)$}
\psfrag{k4c4}[c][][1]{\small$(0,1)$}
\psfrag{k5c5}[c][][1]{\small$(1,1)$}

		\centerline{\includegraphics[scale=0.5]{P5.eps} }
		
		\vspace{-1.5cm}
\caption{\label{fig:P5}Top: The GSS scheme when the graph is a $P_5$. The sums are modulo $2$. Bottom: The same protocol  for $s=1$ and some particular values of $k_i$'s. }
\end{figure}

\paragraph{Recovering the secret.} Since every player has the secret encrypted using the keys of his neighbours, it comes that each player together with his neighbours can recover the secret. So, in the  example given in Figure \ref{fig:P5}, any superset of the following sets of players is authorised: $\{1,2\}, \{2,3,4\}, \{4,5\}$. 
But not all the authorised sets are of that kind:  if the players $1, 3$ and $5$ add up their encrypted secret, the resulting bit is $c_{1}+c_{3}+c_{5} = 3s+2k_2+2k_4 =s \mod 2$, so the set $\{1,3,5\}$ is also an authorised set. 
More generally, we consider the following sets of players and we show that they are authorised sets.

\begin{dfn}
\label{independant}
Given a graph $G=(V,E)$, a set $B \subseteq V$ is \emph{c-accessing} if and only if:
\begin{align}
\exists D \subseteq B &, 
\left\{
\begin{array}{c}
D \cup Odd(D) \subseteq B \\
|D| = 1 \bmod 2
\end{array}
\right.
\end{align}
\end{dfn}

For a given graph $G=(V,E)$, let  $B\subseteq V$ be a c-accessing set and let $D\subseteq B$ such that $|D|=1 \mod 2$ and $Odd(D)\subseteq B$. In the following, we show that the players in $B$ can recover the secret by computing $\sum_{i\in D} c_i + \sum_{j\in Odd(D)} k_j\mod 2$.

\begin{eqnarray*}
\sum_{i\in D} c_i + \sum_{j\in Odd(D)} k_j &=& \sum_{i \in D} \left(s + \sum_{{i'} \in \mathcal{N}( {i})}k_{i'} \right) + \sum_{j \in Odd(D)}k_j\\
 &=&|D|.s +\sum_{i \in D, {i'} \in \mathcal{N}( {i})}k_{i'}  + \sum_{j \in Odd(D)}k_j\\
  &=&|D|.s +\sum_{i' \in V}|\mathcal N(i')\cap D|.k_{i'}  + \sum_{j \in Odd(D)}k_j\\
  &=& s\mod 2
\end{eqnarray*}

\noindent As a consequence:

\begin{property}\label{prop:AccSufficient}
Given a graph $G=(V,E)$, any set $B\subseteq V$ of players which is {c-accessing} is an authorised set.
\end{property}

In the rest of the section, we show that any set which is not c-accessing is  forbidden. As a consequence, the c-accessible sets provide a characterisation of the authorised sets. It also proves that for any graph, the protocol is \emph{perfect}, i.e. any set of players is either authorised or forbidden. 

The proof is using the following characterisation of the sets which are not c-accessing, proved in \cite{JMP11}, stating that a set is not c-accessing if and only if it is oddly-dominated by a subset of its complement:

\begin{lem}[\cite{JMP11}]\label{lem:setC}
Given a graph $G=(V,E)$, $B\subseteq V$ is not c-acessing if and only if \begin{equation}\exists C\subseteq \overline B, Odd(C)\supseteq B\end{equation}
\end{lem}

\begin{thm}
\label{AccNecessary}
Given a graph $G=(V,E)$, any set $B\subseteq V$ of players which is not c-accessible is forbidden.
\end{thm}

\begin{proof}
Let $B$ be a set of players which is not c-accessible. The players in $B$ share collectively the following bits: $\big((k_i, c_i)\big)_{i\in B}$. Let $k_B=(k_i)_{i\in B}$ and $c_B=(c_i)_{i\in B}$. We want to show that the players in $B$ have no information about the secret, i.e. $P(s | k_B,c_B)=P(\bar s|k_B,c_B)$. 

Notice that $P(\bar s | k_B,{c_B})=P(s|k_B,\overline{c_B})$ since $\forall i, \overline{c_i}=\bar s+\sum_{j\in \mathcal N(i)} k_i$. Moreover, $P(s|k_B,c_B)=\frac{P(s,c_B|k_B)}{P(c_B|k_B)}$. 
According to Lemma \ref{lem:setC}, there exists $C\subseteq \overline B$ s.t. $Odd(C)\supseteq B$, i.e. $\forall i\in B, |\mathcal N (i)\cap C| = 1\mod 2$.  For any $i\in B$, 

\begin{eqnarray*}
 \overline{c_i}&=& 1+ s+\sum_{j\in \mathcal N(i)}k_i \mod 2\\
&=&1+s+\sum_{j\in \mathcal N(i)\cap  C}k_i +\sum_{j\in \mathcal N(i)\cap  \overline C}k_i \mod 2\\
& = &1+s+|\mathcal N(i)\cap C|+\sum_{j\in \mathcal N(i)\cap  C}{(k_i-1)}+\sum_{j\in \mathcal N(i)\cap  \overline C}k_i \mod 2\\
& = &s+\sum_{j\in \mathcal N(i)\cap  C}\overline{ k_i}+\sum_{j\in \mathcal N(i)\cap  \overline C}k_i \mod 2\\
\end{eqnarray*}

As a consequence, $P(\overline{c_B}|k_B,k_C)=P(c_B|k_B,\overline{k_C})$ and $P(s,\overline{c_B}|k_B,k_C)=P(s,c_B|k_B,\overline{k_C})$. Hence,

\begin{eqnarray*}
P(\bar s|k_B,c_B)&=&P(s|k_B,\overline{c_B})\\
&=&\frac{P(s,\overline{c_B}|k_B)}{P(\overline{c_B}|k_B)}\\
&=&\frac{\sum_{k_C\in \{0,1\}^{C}} P(k_C)P(s,\overline{c_B}|k_B,k_C)
}{\sum_{k_C\in \{0,1\}^{C}} P(k_C)P(\overline{c_B}|k_B,k_C)}\\
&=&\frac{\sum_{k_C\in \{0,1\}^{C}} \frac1{2^{|C|}}P(s,{c_B}|k_B,\overline{k_C})
}{\sum_{k_C\in \{0,1\}^{C}} \frac1{2^{|C|}}P({c_B}|k_B,\overline{k_C})}\\
&=&\frac{\sum_{k_C\in \{0,1\}^{C}} \frac1{2^{|C|}}P(s,{c_B}|k_B,{k_C})
}{\sum_{k_C\in \{0,1\}^{C}} \frac1{2^{|C|}}P({c_B}|k_B,{k_C})}\\
&=&\frac{P(s,{c_B}|k_B)}{P({c_B}|k_B)}\\
&=&P(s|k_B,c_B)
\end{eqnarray*}
\end{proof}

Property \ref{prop:AccSufficient} and Theorem \ref{AccNecessary} provide a characterisation of accessible sets.
Moreover, any set which is not c-accessible cannot learn any information about the secret.

\section{Graph state quantum secret sharing}\label{sec:gs-qss}

Secret sharing with graph states (GS-QSS) has been introduced by Markham and Sanders in \cite{MS08}. For a given graph of order $n$, the secret -- which can be either classical or quantum -- is encoded into the n-partite quantum state described by the graph, the so called graph state \cite{HEB04}. Then each of the $n$ players receives one qubit of the quantum state. 

It has been shown in \cite{JMP11} that, when the secret is a classical bit, every subset of players is either authorised or forbidden, and that the authorised sets are the c-accessing sets in the corresponding graph i.e., the subsets $B$ of players such that $\exists D\subseteq B$, $|D|=1\mod 2$ and $Odd(D)\subseteq B$. As a consequence, 

\begin{property}
Given a graph $G$, when the secret is a classical bit, the accessing structure of the GS-QSS scheme characterised by $G$ coincides with the accessing structure of the GSS characterised by $G$. \end{property}

Thus, when the secret is classical, any GS-QSS protocol can be simulated by a GSS protocol, which is  simple classical protocol that consists in sending each player only two  bits. It shows, when the secret is classical, there is no benefit to use a  graph state quantum protocol rather than a fully classical protocol.

When the secret is an arbitrary quantum state,  an interesting reduction to the classical secret case has been shown in \cite{MS08}: given a graph $G$, a set of players is authorised for a quantum secret in the GS-QSS protocol characterised by $G$ if and only if this set of players is authorised in both $(i)$ the GS-QSS protocol characterised by $G$ when the secret is classical, and $(ii)$ the GS-QSS protocol  characterised by $\overline G$ when the secret is classical. As a consequence:

\begin{property}
When the secret is an arbitrary quantum state, the authorised  sets of players in a GS-QSS scheme for a graph $G$  are those which are authorised in  the two particular instances $G$ and $\overline  G$ of the GSS protocol. 
\end{property}

In \cite{JMP11}, it has been proven that a set $B$ of players is c-accessible in both $G$ and $\overline G$ if and only if $B$ is c-accessible in $G$ and $\overline B$ is not c-accessible in $G$. As a consequence, the accessing structure of a GSS protocol provides a full characterisation of the authorised sets in the corresponding GS-QSS protocol:

 \begin{property}
When the secret is an arbitrary quantum state, a set $B$ of players is authorised  in a GS-QSS scheme for a graph $G$  if and only if $B$ is authorised and $\overline B$ is forbidden in the GSS scheme for $G$. 
\end{property}

\section{Conclusion}

In this note, we characterise the accessing structure of a simple graph-based secret sharing (GSS) protocol. Moreover, we show that this simple protocol is strongly related to the study of the graph-state quantum secret sharing (GS-QSS) protocols. We point out that, when the secret is classical, any GS-QSS scheme can be simulated by a GSS protocol. Moreover, when the secret is an arbitrary quantum state, the accessing structure of a GSS provides a full characterisation of the authorised sets in the corresponding GS-QSS: an authorised set in the quantum case is a set which is authorised in the classical case and such that its complement set is forbidden in the classical case.

\section*{Acknowledgments}
This work is partially supported by the CNRS PEPS project GraphIQ.


\begin{thebibliography}{1}
\bibitem{B}
G.R. Blakley,
\newblock  \emph{Safeguarding cryptographic keys.}
\newblock {\sl AFIPS Conference Proceedings.} 48 (1979) 313Ð317.


\bibitem{CG99}
R. Cleve, D. Gottesman, and H.-K. Lo, 
\newblock \emph{How to Share a Quantum Secret}
 \newblock {\sl Phys. Rev. Lett.} 83, 648-651 (1999). Also quant-ph/9901025.

\bibitem{G00}
D. Gottesman, 
\newblock \emph{On the Theory of Quantum Secret Sharing.}
 \newblock {\sl Phys. Rev. A} 61, 042311 (2000) (8 pages). Also quant-ph/9910067.

\bibitem{HEB04} M. Hein, J. Eisert, and H. J. Briegel.
\newblock \emph{Multi-party entanglement in graph states.}
\newblock {\sl Physical Review A}, 69, 2004. quant-ph/0307130.

\bibitem{JMP11}
J.~Javelle, M.~Mhalla, and S.~Perdrix.
\newblock\emph{New Protocols and Lower Bound for
Quantum Secret Sharing with Graph States}
\newblock  arXiv:1109.1487 (2011).


\bibitem{KMMP}
E.~Kashefi, D.~Markham,  M.~Mhalla,  and S.~Perdrix.
\newblock {Information Flow in Secret Sharing Protocols.}
 \newblock {\sl DCM 2009: Elec. Proc. Theor. Comp. Sci. 9, 87 (2009).}



\bibitem{MS08}
D.~Markham and B.~C. Sanders.
\newblock \emph{Graph states for quantum secret sharing.}
\newblock {\sl Physical Review A} 78, 042309, 2008.


\bibitem{Shamir}
A. Shamir. 
\newblock \emph{How to share a secret}.
\newblock {\sl Communications of the ACM 22 }(11): 612Ð613, 1979.


\end{thebibliography}
\end{document}